\documentclass[conference]{IEEEtran}
\IEEEoverridecommandlockouts

\usepackage[no-math]{fontspec}
\defaultfontfeatures{Ligatures=TeX,Scale=MatchLowercase}
\usepackage{cite}
\usepackage{amsmath,amssymb,amsfonts}
\usepackage{algorithmic}
\usepackage{graphicx}
\usepackage{textcomp}
\usepackage{xcolor}
\usepackage{bm}
\usepackage{subfigure}
\usepackage{amsthm}
\usepackage[ruled]{algorithm2e}
\usepackage{booktabs}
\usepackage{multirow}
\usepackage{multicol}
\usepackage{makecell}
\usepackage{array}
\usepackage{cleveref}

\def\BibTeX{{\rm B\kern-.05em{\sc i\kern-.025em b}\kern-.08em
    T\kern-.1667em\lower.7ex\hbox{E}\kern-.125emX}}

\newtheorem{remark}{Remark}
\newtheorem{proposition}{Proposition}
\newtheorem{corollary}{Corollary}

\graphicspath{{./Figures/}}
 
\begin{document}

\title{Evaluating Power Flow Manifold from Local Data around a Single Operating Point via Geodesics\\
}

\author{Qirui Zheng, Dan Wu,~\IEEEmembership{Member,~IEEE,} Franz-Erich Wolter, and Sijia Geng,~\IEEEmembership{Member,~IEEE}
\thanks{Qirui Zheng and Dan Wu are with the School of Electrical and Electronic Engineering, Huazhong University of Science and Technology, Wuhan, China (e-mail: \{zqr\_2024, danwuhust\}@hust.edu.cn).}
\thanks{Franz-Erich Wolter is with the Department of Electrical Engineering and Computer Science,
Leibniz University Hannover, Hannover, Germany (e-mail: few@gdv.uni-hannover.de).}
\thanks{Sijia Geng is with the Department of Electrical and Computer Engineering,
Johns Hopkins University, Baltimore, MD 21218, USA (e-mail: sgeng@jhu.edu).}
\thanks{Correspondence: Dan Wu.}
}



\maketitle

\begin{abstract}
The widespread adoption of renewable energy poses a challenge in maintaining a feasible operating point in highly variable scenarios. 
This paper demonstrates that, within a feasible region of a power system that meets practical stability requirements, the power flow equations define a smooth bijection between nodal voltage phasors (angle and magnitude) and nodal active/reactive power injections. Based on this theoretical foundation, this paper proposes a data-based power flow evaluation method that can imply the associated power flow manifold from a limited number of data points around a single operating point. 
Using techniques from differential geometry and analytic functions, we represent geodesic curves in the associated power flow manifold as analytic functions at the initial point. 
Then, a special algebraic structure of the power flow problem is revealed and applied to reduce the computation of all higher-order partial derivatives to that of the first-order ones. 
Integrating these techniques yields the proposed data-based evaluation method, suggesting that a small number of local measurements around a single operating point is sufficient to imply the entire associated power flow manifold. Numerical cases with arbitrary directional variations are tested, certifying the efficacy of the proposed method.
\end{abstract}

\begin{IEEEkeywords}
Power flow manifold, geodesic equation, differential geometry, data-based method, analytic function
\end{IEEEkeywords}

\section{Introduction}
 Reducing carbon emissions in the electricity sector is a major measure to achieve carbon neutrality, which can be realized by replacing fossil-fuel-based generation with renewable energy. However, widespread installation of renewable generation can pose many challenges to the stable and reliable operations of modern power grids \cite{anderson2025ten}. One major challenge confronting power system operators is to ensure that the system maintains a feasible operating point in highly variable and intermittent renewable generation scenarios. 
Traditional routines based on a limited number of pre-examined typical operation modes and patterns are becoming less effective in evaluating increasingly diversified situations \cite{tri_sec}. New methods that can cover a wide, if not global, region of feasible power variations are therefore required. 

 The feasibility region, also known as the solvability region, of the power flow problem is usually defined as a set of power injection vectors for which the AC power flow equations admit at least one solution. It is related to the classical voltage stability problem, in which the singularity condition of the power flow Jacobian matrix describes the solvability boundary \cite{VS_margins}. 
Early attempts focused on the existence (sometimes uniqueness) conditions of the power flow solutions \cite{QV,radialS,sol,fes_sol}, leaving the estimation of the solvability region implicit. The breakthrough came with the reformulation of the power flow problem as a fixed-point problem using Banach contraction \cite{Linear_App}. Several explicit inner approximation methods based on various fixed-point theories were then proposed \cite{Inner_App,Feas_Sets,cui2019solvability}. They usually construct a polytopic or convex region that lies within the solvability region and serves as a certificate of feasibility. Although theoretically rigorous, existing inner approximation methods are model-sensitive and overly conservative, especially when estimating high-dimensional regions. 

 With the advancement of modern data science and machine learning, data-driven methods have been applied to power flow \cite{DL} and associated voltage stability assessment problems \cite{datadrive1_b,datadrive1_b2,datadrive2,datadrive1,datadrive3improve,TRANs_DL2,TRANs_DL3}. Typically, these problems are treated as a boundary-learning problem with a trained neural-network classifying if a given power injection profile is within the feasibility region \cite{DL}. The supervised learning paradigm is commonly adopted here. However, real-world power system measurements are hardly available to serve as an adequate sample pool. The scarcity of extreme events near the boundary further makes the real-world data set not representative of all scenarios. Hence, the majority existing data-driven methods feed synthetic or simulated power flow data into their training processes \cite{DL,datadrive1_b,datadrive1_b2,datadrive2,datadrive1,datadrive3improve,TRANs_DL2,TRANs_DL3}, which, in return, build trained models that are system-specific and intransferable to different network structures. Moreover, the poor interpretability of neural networks is responsible for the lack of guarantees or theoretical certificates, which prevents their operational deployment. 

To address the conservativeness of analytical methods, expensive data acquisitions, and poor interpretability of existing data-driven methods, in this article we extend our previous work \cite{WU2024110716} by leveraging differential geometry, power flow structure, and analytic function techniques to propose a novel data-based approach that is physically meaningful, geometrically interpreted, and can evaluate the associated power flow manifold from a limited number of measurements around a single operating point. 

The associated power flow manifold, which will be rigorously defined in Section~\ref{section:evaluating}, describes the region in which the system can continuously evolve to without traversing the singular boundary. This manifold is shown to be a proper length space that admits a shortest path between any pair of points. In the power engineering context, such a path is physically meaningful only when it does not touch the singular boundary (i.e., interior), which implies a geodesic. Hence, \textit{how to utilize geodesic curves} is at the heart of our proposed approach to evaluate the associated power flow manifold. 

A straightforward way is to solve the geodesic equation step-by-step, extending the local geodesic curve to consecutive neighborhoods. However, this method requires one to evaluate numerous points along the geodesic curve, which inevitably requires information on the metric tensor and the associated Christoffel symbols at points in those neighborhoods. With a known power system model, such information can be obtained from model evaluation. However, in a model-free context, sufficient measurements along the geodesic curve are indispensable to reproduce this information. Such a data request would easily become formidable when a high-dimensional power flow manifold is investigated with countless geodesic directions induced by stochastic variations of renewable energy. 

To address this challenge, in this paper, we propose a data-based method that can evaluate the associated power flow manifold from a limited amount of data acquired around a single operating point. The contents are arranged as follows. Section~\ref{section:model} introduces some basic concepts from differential geometry that are used in the paper and derives the analytic representation of geodesic curves. The intrinsic geometric properties and benign behaviors of analytic functions ensure a global representation of any geodesic curve from local information at the initial point. It substantially reduces the data requirement from a cloud of points globally over the manifold to a single point in the manifold. To further reduce the information required for higher-order derivatives, a thorough investigation of algebraic structures of the power flow map is conducted in Section~\ref{section:derivatives}, revealing that \textit{the first-order derivatives are all we need}. Upon integration of these techniques, the main theme of Section~\ref{section:evaluating} is to articulate the requirements for local measurements and the proposed data-based evaluation approach. The proposed method is numerically tested in Section~\ref{sec:simu} under various conditions and different data ranges. 

A key methodological takeaway is that, \textit{a limited amount of local measurements around a single operating point is sufficient to imply the associated power flow manifold}. We present proof-of-concept for the computation using numerical examples, showing that if we have a tiny neighborhood in phase-voltage space of a given operating point with $n$ linearly independent vectors, we can approximate the essential global geometric data of the aforementioned Riemannian manifold and its singular boundary:


1) Firstly, using the aforementioned point data, we can approximate the power flow Jacobian matrix.

2) Using the Jacobian matrix, we can approximate the partial derivatives of any order of the power flow map at the operating point. 
 
3) The higher-order derivatives of the power flow map will provide all partial derivatives of the Riemannian metric tensor and its inverse, and thus give us all the data describing Taylor series of any order for geodesics emanating from the operating point on the Riemannian manifold for any given initial direction.

4) Furthermore, using the Taylor approximation of this geodesic, we can compute a sufficiently precise (rational) Padé approximation of the geodesic, and using its related singular value, i.e., where the denominator of the rational Padé approximation is zero, to yield the position where the geodesic first hits the singular boundary of the bordered Riemannian manifold. 
 
The validity of the approach described above is proven by theoretical arguments and confirmed through numerical computations for benchmark power systems, where we observe that values computed via the approximations 1) - 4) agree with the respective model-based results in these sample cases.

\section{Analytic Representation of Geodesic Curves}\label{section:model}

In this section, we briefly introduce some necessary concepts in differential geometry used in the paper. A detailed introduction can be found in \cite{manifold,WU2024110716}. Readers familiar with those concepts can directly jump to Section~\ref{section:derivatives}. Throughout this paper, the Einstein summation convention\footnote{If an index appears twice in a term with one superscript and one subscript, it is summed over all possible values of that index. For example, $a^ib_i \triangleq \sum_i a^ib_i$, where the superscript $i$ represents the index (usually denotes contravariant), not the exponent.} is adopted unless otherwise stated. 


\subsection{Basic Quantities in Differential Geometry}\label{subsection:basic}

\newtheorem{definition}{Definition}

\begin{definition}\label{def:map}
    (Smooth Parameterization) Let $X$ and $Y$ be two smooth manifolds, $\dim X=\dim Y=n$, if there exists a smooth map $r:X \to Y$ such that $r$ is bijective and $\forall x\in X$, $r$ is a local immersion (i.e., the Jacobian matrix of $r$ with respect to x, namely, $Jr_x$, is injective), then X is a smooth parameterization of Y.
\end{definition}
Consider the smooth manifold $Y$ in Definition ~\ref{def:map}, which permits a smooth parameterization from a smooth manifold $X$ and the corresponding smooth map $r$. We define the bases $\bm r_i$ as the first-order partial derivatives,
\begin{equation}\label{eq:basis}
     \bm r_i=\frac{\partial r}{\partial x^i},
\end{equation}%
where $x^i$ is the $i$-th component of $x \in X$. 


Then, the metric tensor is obtained as,
\begin{equation}\label{eq:g}
    g_{ij}=\langle\bm r_i,\bm r_j\rangle,
\end{equation}%
where $\langle\bm r_i,\bm r_j\rangle$ is the inner product of $\bm r_i$ and $\bm r_j$. 

The contravariant metric tensor is the inverse of the metric tensor which always exists according to Definition~\ref{def:map} since the map $r$ is bijective and nonsingular,
\begin{equation}\label{eq:ginv}
    (g^{ij})=(g_{ij})^{-1}.
\end{equation}

As the basis $\bm r_i$ varies along the parameter $x^j$, we have the rate of change of the basis, 
\begin{equation}
    \bm r_{ij}=\frac{\partial \bm r_{i}}{\partial x^j}.
\end{equation}

The projection component of $\bm r_{ij}$ on the basis $\bm r_m$ is given by the Christoffel symbols of the second kind $\Gamma^m_{ij}$, 
   \begin{equation}\label{eq:cris}
    \Gamma^m_{ij}=\langle \bm r_{ij},\bm r_{k} \rangle g^{mk}.
\end{equation} 

The partial derivatives of $g^{mn}$ and $\Gamma^m_{ij}$ are then given by, 
\begin{subequations}\label{eq:partial_cg}
\begin{align}
    \frac{\partial g^{mn}}{\partial x^k}&=-g^{im}\Gamma^n_{ik}-g^{in}\Gamma^m_{ik},\\
    \frac{\partial\Gamma^m_{ij}}{\partial x^k}&=\langle\bm r_{ijk},\bm r_{n}\rangle g^{mn}+\bm \langle \bm r_{ij}, \bm r_{nk}\rangle g^{mn}+\langle\bm r_{ij},\bm r_{n}\rangle\frac{\partial g^{mn}}{\partial x^k},
\end{align}
\end{subequations}
where $\bm r_{ijk}$ is the third order partial derivative of $y=r(x) \in Y$ with respect to $x^i$, $x^j$, and $x^k \in X$.

\subsection{The Geodesic Equation}\label{subsec:geo}


A geodesic curve in a non-flat manifold is an analogy of a straight line in a flat space. There is no tangential acceleration along the geodesic curve, which is described by the geodesic equation,
\begin{equation}\label{eq:geodesic}
    \frac{d^2 x^m}{d \lambda^2}+\Gamma^m_{ij}\frac{d x^i}{d \lambda}\frac{d x^j}{d \lambda}=0,
\end{equation}
where $\lambda$ is a canonical parameterization that is proportional to the arc length of the geodesic curve. This is a set of second-order nonlinear differential equations.


\subsection{Analytic Representation of Geodesic Curves}

\newtheorem{theorem}{Theorem}
\newtheorem{lemma}{Lemma}
\begin{lemma}\label{lemma:geo}
Consider the smooth manifold $Y$ in Definition ~\ref{def:map}. $X$ is the smooth parameterization of $Y$.  The solution to \eqref{eq:geodesic} is a smooth function $x^m$ with respect to $\lambda$, which, thus, admits a Taylor series expansion. The $k$-th order coefficient $\Omega_k$ of the Taylor series expansion of the geodesic curve at $\lambda = 0$ is determined by the partial derivatives of $r$ whose orders are no greater than $k$ and by the velocity $\dot{x}^j$ for all $j$, namely,  
\begin{subequations}
 \begin{align}
        x^m(\lambda) &= x^m(0)+\dot{x}^m(0) \times \lambda+\frac{\ddot{x}^m(0)}{2!} \times (\lambda)^2+\dots, \nonumber\\
        &=  \sum\Omega_k \times (\lambda)^k,\label{eq:analytic}\\
     \Omega_k &= \Omega_k \big(\bm r_{i_{11}},\bm r_{i_{21} i_{22}},\dots,\bm r_{i_{k1} \dots i_{kk}},\dot{x}^j \big),\label{eq:taylor}
    \end{align}   
\end{subequations}
where $\dot{x}^m = dx^m/d\lambda$; $\Omega_k$ is the $k$-th order coefficient of the Taylor series; 
the subscripts $i_{11}, \dots, i_{kk}$ are index variables taking values over the set of integers from $1$ to $n$; $\bm r_{i_{k1}, \dots, i_{kk}}$ represent the $k$-th order partial derivatives of $y=r(x)$ with respect to $x$; $(\lambda)^k$ represents the $k$-th order monomial of $\lambda$. 


\end{lemma}

The proof of Lemma~\ref{lemma:geo} is by induction and is given in Appendix~\ref{append:1}.

Note that estimating higher-order derivatives $\bm r_{i_{k1} \dots i_{kk}}$ directly from data would normally require an enormous amount of observations. Fortunately, in the power system context, these derivatives of the power flow map can be fully recovered from the first-order derivatives, substantially reducing the data acquisition burden, as will be detailed in Section \ref{section:derivatives}. 

\subsection{Rational Representation of Geodesic Curves}\label{subsec:pade}
The Taylor series may only have a finite radius of convergence. To address this problem, we further represent the geodesic curve as a rational function by the Padé approximation, as shown below,
\begin{equation}\label{eq:Pade}
    x^m(\lambda)={\sum_{n=0}^\infty a_n \times (\lambda)^n}/{\sum_{l=0}^\infty b_l \times (\lambda)^l}.
\end{equation}
The coefficients $a_n$ and $b_n$ of the Padé approximant can be derived from the Taylor series by equating \eqref{eq:analytic} to \eqref{eq:Pade} and matching up the coefficients of the same order below,
\begin{equation}\label{eq:calpade}
\bigg(\sum_{l=0}^{\infty} b_l \times (\lambda)^l \bigg) \bigg(\sum_{k=0}^{\infty}\Omega_k \times (\lambda)^k \bigg) =\sum_{n=0}^\infty a_n \times (\lambda)^n.
\end{equation}
A detailed formula of Padé coefficients can be found in \cite{Holomorphic}.
\subsection{Shortest Path on Riemannian Manifold}
\begin{lemma}\label{lemma:distance}
Let $(Y,g)$ be a connected Riemannian manifold with boundary $\partial Y$, let $d_Y$ be the intrinsic distance from the metric tensor $g$. If  $(Y,d_Y)$ is compact or closed in a complete ambient manifold, $(Y,d_Y)$ is a proper length space, i.e., $\forall$ $y_0,y_1\in Y$, $\exists$ a shortest path $\gamma(t)\subset Y:\gamma(0)=y_0,\gamma(1)=y_1$. If $\gamma$ lies entirely in the interior of $Y$ except possibly at its endpoints $y_0$ and $y_1$, it is a geodesic curve.
\end{lemma}
The proof can be found in page 8 in \cite{Franz_master}. This Lemma implies that there exists a shortest path between two points. It can be either a geodesic curve joining them directly (interior) or a curve concatenated by several segments of geodesics and paths on the boundary. More complicated scenarios can be found in \cite{Franz_master, Franz_phd}.  In the power system context, we are particularly interested in where the system can evolve to without traversing the singular boundary. Hence, the interior geodesic curves in certain power flow manifolds are of particular importance, which is the main topic of the following sections.

\section{Power Flow Map in Differential Geometry}\label{section:derivatives}
In this section, we focus on properly defining the power flow map and connecting it to its physical interpretation. We restrict our attention to certain manifolds that are commonly considered in power engineering but are seldom well-defined. Then, we show that the higher-order partial derivatives of the power flow map can be expressed in terms of the first-order derivatives and the current operating point. This result substantially simplifies the complicated higher-order derivatives introduced in Section~\ref{section:model}. We temporarily drop the Einstein summation convention in this section to accommodate a common power engineering index convention. 

\subsection{The Power Flow Map and The Associated Manifold}
Consider a power grid with $N$ nodes. Without loss of generality, assume the first $N_l$ buses are PQ buses where both active and reactive power injections are specified. The buses from $N_l+1$ to $N_l+N_g$ are PV buses where active power injections and voltage magnitudes are specified. And the last bus is the slack bus, its voltage magnitude $V_N$ and angle $\delta_N$ are specified.
The power flow equation is shown below
\begin{subequations}\label{eq:cpq}
\begin{align}
V_i\sum_{j=1}^NV_j(G_{ij}\cos\delta_{ij}+B_{ij}\sin\delta_{ij}) &= P_i\label{eq:cp}\\
V_i\sum_{j=1}^NV_j(G_{ij}\sin\delta_{ij}-B_{ij}\cos\delta_{ij}) &= Q_i\label{eq:cq}
\end{align}
\end{subequations}%
where $P_i$ and $Q_i$ are the active and reactive power injections at the $i$-th bus, respectively; $V_i$ is the voltage magnitude of the $i$-th bus, $\delta_{ij} = \delta_i-\delta_j$ represents the phase angle difference between the $i$-th node and the $j$-th node; $G_{ij}$ and $B_{ij}$ are the $(i,j)$-th elements of the bus conductance and susceptance matrices. 
\begin{definition}\label{def:PFmap}
(Power Flow Map) Consider the power grid defined above. Let $x = (\delta, V) \in X$ where $\delta$ comprises $\delta_1$ to $\delta_{N-1}$ and $V$ comprises $V_1$ to $V_{N_l}$. Let $y=(P, Q_l) \in Y$ where $P$ comprises $P_1$ to $P_{N-1}$ and $Q_l$ comprises $Q_1$ to $Q_{N_l}$. Then, the power flow equation \eqref{eq:cpq} defines a map $r$ from $X$ to $Y$ such that
    \begin{equation}\label{eq:PF_map}
        r:X \to Y,~r(x)=y, 
    \end{equation}
    where $r$ includes \eqref{eq:cp} from bus-$1$ to bus-$(N_l+N_g)$ and \eqref{eq:cq} from bus-$1$ to bus-$N_l$.
\end{definition}

Given the non-unique nature of power flow solutions, the power flow map is generally not a bijection. Consequently, X does not constitute a parameterization of Y, and the differential geometric methods outlined in Section \ref{section:model} cannot be applied directly. However, by imposing appropriate constraints—such as those specified in Definition~\ref{def:subm}, which are consistent with practical engineering considerations—a smooth parameterization can be obtained.

\begin{definition}\label{def:subm}
     (Associated Manifold of $\xi$) Consider a feasible operating point $\xi=(x_0,y_0)$ that satisfies the power flow map $r$ within engineering constraints. The associated manifold $\mathcal{X}_\xi \subseteq X$ of $\xi$ is the set such that $\forall x \in \mathcal{X}_\xi$, $\exists$ a continuous path $\gamma:[0,1]\to X$ such that $\gamma(0)=x_0$, $\gamma(1)=x_1$, and $\forall t\in[0,1]$,  the Jacobian matrix $Jr_{x(t)}$ is invertible. Furthermore, for any $\delta_k$ of $x$, $\exists$ $L_k\in(0, 2\pi)$ and $b_k\in \mathbb{R}$ such that, 
     \begin{equation}\label{eq:2pilength}
    \delta_k \in \bigl[\,b_k,\; b_k + L_k\,\bigr].
    \end{equation}
    Let $\mathcal{Y}_\xi$ be the induced manifold in Y such that, 
\begin{equation}
    \mathcal{Y}_\xi=\{y=r(x)|x\in\mathcal{X}_\xi\}.
\end{equation}
\end{definition}

\begin{definition}\label{def:boundary}
    (Associated Singular Boundary) Let $\partial\mathcal{X}_\xi$ be the boundary of $\mathcal{X}_\xi$ such that $\forall x \in \partial\mathcal{X}_\xi$,  $\exists$ a continuous path $\gamma:[0,1]\to X$ such that $\gamma(0)=x_0$, $\gamma(1)=x$, and $\forall t\in[0,1)$,  the Jacobian matrix $Jr_{x(t)}$ is invertible, while at $t=1$, $Jr_{x(1)}$ is singular. Furthermore, for any $\delta_k$ of $x$, $\exists$ $L_k\in(0, 2\pi)$ and $b_k\in \mathbb{R}$ such that \eqref{eq:2pilength} holds. 
    Let $\partial\mathcal{Y}_\xi$ be the boundary of $\mathcal{Y}_\xi$ such that 
    \begin{equation}
        \partial\mathcal{Y}_\xi=\{y \in Y |y=r(x),x\in\partial\mathcal{X}_\xi \}.
    \end{equation}
\end{definition}

\begin{proposition}\label{proposition:propermap}
 The power flow map $r$ defines a proper smooth map from $\mathcal{X}_\xi$ to $\mathcal{Y}_\xi$.
 \begin{proof}
     According to \cite{lesieutre2015:efficient,wu2017phd}, the power flow model in rectangular coordinates can be formulated as an equivalent family of quadratic ellipsoids under very mild conditions. A quadratic ellipsoidal map is a smooth proper map. Furthermore, under the angle constraints of \eqref{eq:2pilength}, polar coordinates can be mapped onto rectangular coordinates via a smooth bijection. Thus, the power flow map yields a smooth proper map from $\mathcal{X}_\xi$ to $\mathcal{Y}_\xi$.
 \end{proof}
\end{proposition}

 \begin{proposition}\label{pro:sp}
    The power flow map $r$ defines a smooth parameterization from $\mathcal{X}_\xi$ to $\mathcal{Y}_\xi$. 
\begin{proof}
    By Definition~\ref{def:subm}, $\mathcal{X}_\xi$ and $\mathcal{Y}_\xi$ are connected manifolds, and they are local diffeomorphic since $r$ is smooth and the Jacobian matrix $Jr_x$ is invertible. By Proposition~\ref{proposition:propermap}, the power flow map is a smooth proper map from $\mathcal{X}_\xi$ to $\mathcal{Y}_\xi$. Due to the properness, connectedness, and local diffeomorphism properties of the map, it constitutes a smooth parameterization from $\mathcal{X}_\xi$ to $\mathcal{Y}_\xi$.
\end{proof}
 \end{proposition}
Proposition~\ref{pro:sp} indicates that for an associated manifold satisfying \eqref{eq:2pilength}, the differential geometric machinery presented in Section~\ref{section:model} can be applied.

\begin{proposition}\label{pro:pls}
    Let $\bar{\mathcal{X}}_\xi=\mathcal{X}_\xi \cup \partial\mathcal{X}_\xi$, $\bar{\mathcal{Y}}_\xi=\mathcal{Y}_\xi \cup \partial\mathcal{Y}_\xi=r(\bar{\mathcal{X}}_\xi)$, $d_{{\mathcal{Y}}_\xi}$ be the intrinsic distance from the metric tensor $g_{\mathcal{Y}_\xi}$, then $(\bar{\mathcal{Y}}_\xi,d_{{\mathcal{Y}}_\xi})$ is a proper length space.
\begin{proof}
    By Definition~\ref{def:subm} and \ref{def:boundary}, $\bar{\mathcal{X}}_\xi$ is connected and closed in $\mathbb{R}^{2N_l+N_g}$ (a complete ambient space). By the same argument in Proposition~\ref{proposition:propermap}, the power flow map is a smooth proper map from $\bar{\mathcal{X}}_\xi$ to $\bar{\mathcal{Y}}_\xi$. Thus, $\bar{\mathcal{Y}}_\xi$ is a connected manifold closed in $\mathbb{R}^{2N_l+N_g}$. By Lemma~\ref{lemma:distance}, $\bar{\mathcal{Y}}_\xi$ is a proper length space.
\end{proof}
\end{proposition}
Proposition~\ref{pro:pls} implies that there always exists a shortest path between any pair of operating points in the power space defined in Definition~\ref{def:subm}. 
\subsection{The Full Power Flow Map and Its Partial Derivatives}
\begin{definition}\label{def:fullmap}
    (Full Power Flow Map) We construct the full power flow map $\tilde{r}$ by incorporating the extra reactive power balance equations based on \eqref{eq:cq} at PV buses into the power flow map $r$, 
    \begin{equation}\label{eq:full_PF}
        \tilde{r}:X \to \tilde{Y},~\tilde{r}(x) = \tilde{y} = (P, Q),
    \end{equation}%
    where $Q$ comprises $Q_1$ to $Q_{N-1}$; $\tilde{Y}$ is a $(2N_l+2N_g)$-dimensional manifold.
\end{definition}

All derivatives of arbitrary order for the power flow map $r$ are encompassed within the full power flow map $\tilde{r}$. In the following, we will demonstrate a key property of the derivatives of the full power flow map: any higher-order derivative can be expressed in terms of the first-order derivatives and the information at the current operating point.
\begin{lemma}\label{th:derive} 
Considering the full power flow map $\tilde{r}$ defined in Definition~\ref{def:fullmap}, all higher-order partial derivatives of $(P, Q)$ with respect to $(\delta, V) $ can be expressed as rational functions of the first-order derivatives of $(P, Q)$ with respect to $\delta$ and certain $(P,Q,V)$ value at the current operating point $\xi$. As shown in \eqref{eq:rational},
\begin{equation}\label{eq:rational}
    \tilde{\bm r}_{i_1 \dots i_{k}}=A_k \tilde{\bm r}_{j_1}+B_kP_{j_2}+C_kQ_{j_3},
\end{equation}
where $A_k$, $B_k$, and $C_k$ are the rational functions with respect to some voltages $V_{j_4}$, the subscripts $i_{1} \dots i_{k},j_1,j_2,j_3,j_4$ are index variables taking values over the set of integers from $1$ to $2N-1$.

\end{lemma}
The proof of Lemma~\ref{th:derive} is by induction and is given in Appendix~\ref{append:2}.

\section{Evaluating Power Flow Manifold from Local Measurements}\label{section:evaluating}

\begin{theorem}\label{th1}
Consider the power flow map $r$ from $\bar{\mathcal{X}}_\xi$ to $\bar{\mathcal{Y}}_\xi$. Let $\xi=(x_0,y_0)$ be a regular operating point such that $r(x_0) = y_0$. For any point $y \in \bar{\mathcal{Y}}_\xi$, there exists a shortest path $\eta \subset \bar{\mathcal{Y}}_\xi$ between $y_0$ and $y$. If $\eta$ lies entirely in $\mathcal{Y}_\xi$, $\eta$ is geodesic, and the preimage of $y$, namely, $x = r^{-1} (y) \in \bar{\mathcal{X}}_\xi$, can be evaluated purely from the value of $(P,Q,V)$ at $\xi$ and the first-order partial derivatives of the full power flow map $\tilde{r}$ at $\xi$. 

\begin{proof}
According to Proposition~\ref{pro:pls}, for any pair of points in $\bar{\mathcal{Y}}_\xi$, there must exist a shortest path between them. Thus, $\eta$ exists between $y_0$ and $y$. If we further have $\eta \subset \mathcal{Y}_\xi$, the entire path $\eta$, including the endpoint $y$, is in the interior of $\bar{\mathcal{Y}}_\xi$. Hence, $\eta$ is a geodesic curve by Lemma~\ref{lemma:distance}. 

According to Proposition~\ref{pro:sp}, ${\mathcal{X}}_\xi$ is a smooth parameterization of $\mathcal{Y}_\xi$ under the power flow map $r$, which implies that every point in the geodesic curve $\eta$ is well-defined by the geodesic equation \eqref{eq:geodesic}. Thus, the preimage of $y$, namely, $x = r^{-1} (y)$, can be evaluated from the Taylor series of the preimage curve of $\eta$ at the initial point $x_0$, according to Lemma~\ref{lemma:geo}. 

The coefficients $\Omega_k$'s of this Taylor series are described by \eqref{eq:taylor}, which requires the knowledge of certain higher-order partial derivatives of $r$. According to Lemma~\ref{th:derive}, these higher-order partial derivatives can be completely determined by the first-order partial derivatives of $\tilde{r}$ in \eqref{eq:rational}, which concludes the theorem.
\end{proof}
\end{theorem}

\begin{theorem}\label{th2} 
Under the same assumptions of Theorem~\ref{th1}, continuing the path $\eta$ geodesically until it first reaches the boundary at $y_s \in \partial \mathcal{Y}_\xi$. The intrinsic distance $d_{{\mathcal{Y}}_\xi}(y_0, y_s)$ can be evaluated purely from the value of $(P,Q,V)$ at $\xi$ and the first-order partial derivatives of the full power flow map $\tilde{r}$ at $\xi$.

\begin{proof}
    Following the same arguments in the proof of Theorem~\ref{th1}, the preimage of any point $y$ in the geodesic curve $\eta$ except the endpoint $y_s$ can be evaluated by the Taylor series at the initial point $x_0$. At the preimage of $y_s$, namely, $x_s = r^{-1}(y_s)$, the Jacobian matrix $J_{r(x_s)}$ of the power flow map becomes singular, which fails the smooth parameterization in Proposition~\ref{pro:sp}. Thus, the geodesic equation \eqref{eq:geodesic} is not well-defined at $x_s$, which further prevents the Taylor series expansion. To regain the information at the singularity, apply analytic continuation by constructing the $[L/M]$ Pad\'e approximation \eqref{eq:Pade} from the Taylor series. According to the convergence theorem in \cite{stahl1989:convergence,stahl1997:convergence}, as $L \approx M \to +\infty$, the smallest pole on the real axis will converge to $d_{{\mathcal{Y}}_\xi}(y_0, y_s)$. Since the coefficients of Pad\'e approximation can be calculated from the Taylor series in \eqref{eq:calpade}, and the Taylor series is purely determined from the initial point $\xi$ and the first-order partial derivatives, we conclude the theorem.
\end{proof}
\end{theorem}


\begin{definition}
        (Local Linear Patch) Consider a feasible operating point $\tilde{\xi}=(x_0,\tilde{y}_0)$ that satisfies the full power flow map $\tilde{r}$ within engineering constraints. The local linear patch $Y_L$ is the image of the neighborhood $X_L$ of $x_0$ under the map such that the first-order approximation dominates, i.e., 
\begin{equation}
    Y_L=\big\{\tilde{y} \in \tilde{Y} \big| \tilde{y}=\tilde{r}(x_0+h)= \tilde{y}_0 + h^i {\tilde{\bm r}_{i}} + O(||h||^2)\big\},
\end{equation}%
where $x_0+h \in X_L$. 
\end{definition}

\begin{corollary}\label{cor:2}
    Under the assumptions of Theorem~\ref{th1}, using the sample $(P,Q,V)$ of the current operating point $\tilde{\xi}$ and at least $2N_l+N_g$ many linearly independent samples of $(\Delta P,\Delta  Q,\Delta V,\Delta \delta)$ drawn from the local linear patch of $\tilde{\xi}$, we can estimate the associated manifold at any regular point.

\begin{proof}
    Considering the increment of the power flow map, we have
\begin{align}\label{eq:def_partial}
\begin{bmatrix} 
\frac{\partial P}{\partial\delta}  & \frac{\partial P}{\partial V}  \\ 
\frac{\partial Q}{\partial\delta} & \frac{\partial Q}{\partial V}
\end{bmatrix}
\begin{bmatrix} 
d\delta\\ 
dV\\ 
\end{bmatrix}
=\begin{bmatrix}
    dP \\
    dQ\\
\end{bmatrix}.
\end{align}
These incremental quantities $dP$, $dQ$, $d\delta$, and $dV$ can be estimated from the difference between the given operating point and its nearby points. For example, $dP^1 \approx \Delta P^1 = P^1 - P^0$ where $P^1$ is a point near $P^0$. 

Therefore, the Jacobian matrix in \eqref{eq:def_partial} can be approximated from at least ${2N_l+N_g}$ many data points around the given operating point, 
\begin{align}\label{eq:ca_partial}
\begin{bmatrix} 
\frac{\partial P}{\partial\delta}  & \frac{\partial P}{\partial V}  \\ 
\frac{\partial Q}{\partial\delta} & \frac{\partial Q}{\partial V}
\end{bmatrix}
=\begin{bmatrix}
    \Delta P^1 &\cdots&\Delta P^{n} \\
    \Delta Q^1 &\cdots&\Delta Q^{n}\\
\end{bmatrix} \begin{bmatrix} 
\Delta \delta^1 &\cdots&\Delta \delta^{n} \\ 
\Delta V^1&\cdots&\Delta V^{n}\\ 
\end{bmatrix}^{-1},
\end{align}
where $n={2N_l+N_g}$. 

By Theorem~\ref{th1}, the corollary holds.
\end{proof}
\end{corollary}

\begin{corollary}\label{cor:3}
    Under the same data assumptions of Corollary~\ref{cor:2}, we can estimate the distance to the power flow singular boundary along any geodesic curve. 
    \begin{proof}
        Following the same arguments in the proof of Corollary~\ref{cor:2}, Theorem~\ref{th2} implies the corollary. 
    \end{proof}
\end{corollary}

\begin{remark}
    It should also be noted that only incremental quantities are needed for estimating the Jacobian matrix by \eqref{eq:ca_partial}, suggesting that there is no need for a global angle reference at different buses during the observation. In other words, independent measurements at each bus are sufficient for the proposed method.
\end{remark}
 

The overall evaluation routine is summarized in Algorithm~\ref{algorithm:Ray-based}. 
\begin{algorithm}[tb!]
\caption{Geodesic-Inspired Power Flow Manifold Evaluation Based on Local Data}\label{algorithm:Ray-based}
    \renewcommand{\algorithmicrequire}{\textbf{Input:}}
    \renewcommand{\algorithmicensure}{\textbf{Output:}}
\begin{algorithmic}[1]
\REQUIRE $2N_l + N_g$ measurements near the current operating point and the interested power injection profile. 
\ENSURE Voltage solution or infeasible alarm.
 \STATE Approximate the first-order partial derivatives of the power flow map by \eqref{eq:ca_partial}. (only once)
\STATE Construct basis \eqref{eq:basis}, metric tensors \eqref{eq:g} and \eqref{eq:ginv}, and Christoffel symbols \eqref{eq:cris}. (only once)
\STATE Construct higher-order partial derivatives of the power flow map from the first-order derivatives. (only once)
\STATE Choose the interested power injection profile.
\STATE Calculate coefficients of Taylor series. 
\STATE Calculate coefficients of Padé approximant.
\STATE Evaluate the interested power and voltages.
\RETURN Output
\end{algorithmic}
\end{algorithm}

\section{Numerical Simulations}\label{sec:simu}

In this section, we present numerical examples based on two systems, a 4-bus case and a 9-bus system, to demonstrate the efficacy of the proposed approach. The accuracy of the approximation results is further compared when the data originate from patches of different ranges. The system topologies are shown in Fig.\ref{fig:topology}.
The Taylor series computed in both cases is truncated at the 6th order with $[3/3]$ Padé approximation. 

\begin{figure}[tb!]
	\begin{center}
		\subfigure[4-bus system]{\label{fig:top4}\includegraphics[width=0.5\columnwidth]{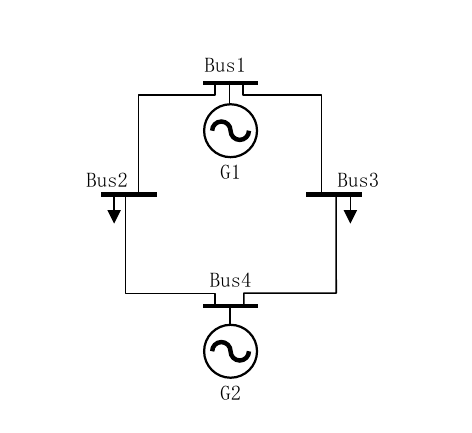}}
        \subfigure[9-bus system]{\label{fig:top9}\includegraphics[width=0.45\columnwidth]{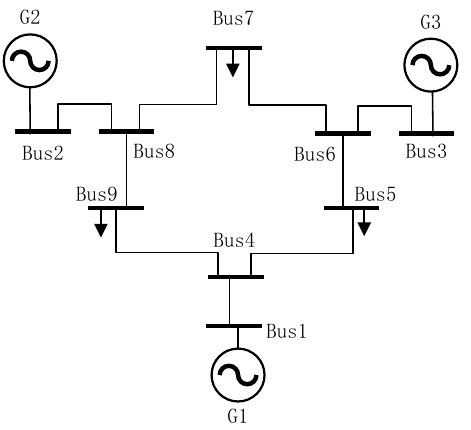}}
	\end{center}
	\caption{Topologies of the studied systems.}
	\label{fig:topology}
\end{figure}

Firstly, local measurements are collected from simulated data with uniform random variations between $\pm 0.05$ p.u. (base power is $100$ MVA) at each nodal power injection. By Corollary~\ref{cor:2}, for the 4-bus system, only 5 data points are used; while the 9-bus system only requires 14 data points. The corresponding voltages are computed from the power flow model. Note that once the simulated data is obtained, we no longer use the power flow model, as the proposed method is completely model-free. In real applications, data from Phasor Measurement Unit (PMU) or Wide Area Management System (WAMS) can be used directly. 

Although the proposed method applies to power flow manifolds in arbitrary dimensions, only 2-dimensional manifolds can be visualized. Thus, we choose to present various 2D power flow surfaces for the testing cases.  
\subsection{Manifold Estimation from Random Local Data}

\begin{figure}[tb!]
	\begin{center}
		\subfigure[]{\label{fig:result_41}\includegraphics[width=0.45\columnwidth]{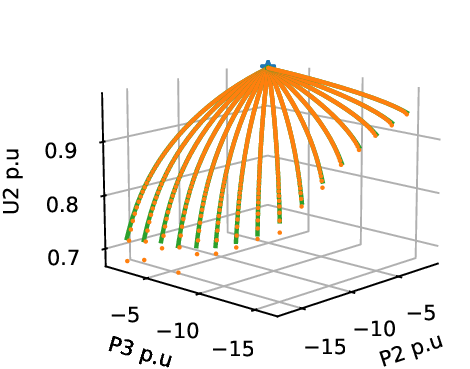}}
        \subfigure[]{\label{fig:result_42}\includegraphics[width=0.45\columnwidth]{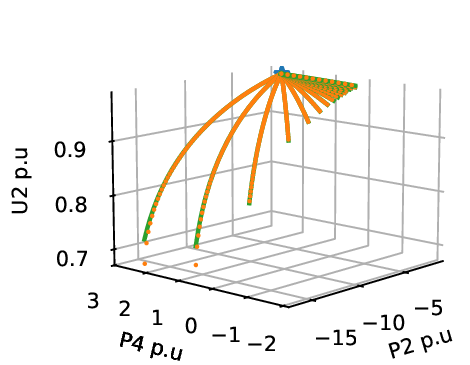}}
	\end{center}
	\caption{Manifold constructions for the 4-bus system (Orange dots: true
manifold, green curves: proposed estimation, 
blue stars: data points used).}
	\label{fig:result4}
\end{figure}

\begin{figure}[tb!]
	\begin{center}		\subfigure{\label{fig:data_4}\includegraphics[width=0.45\columnwidth]{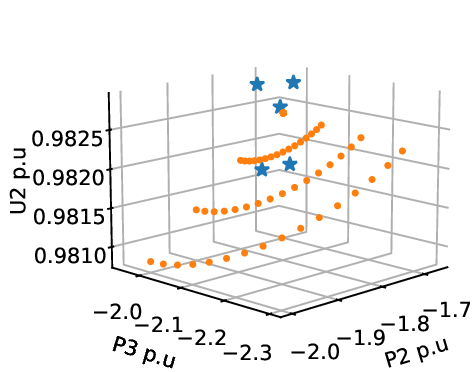}}
    \subfigure{\label{fig:data_4_2}\includegraphics[width=0.45\columnwidth]{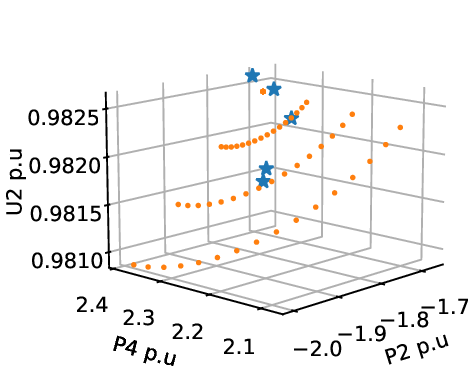}}
	\end{center}
	\caption{Data points used in the 4-bus system (blue stars).}
	\label{fig:data4}
\end{figure}

\begin{figure}[tb!]
	\begin{center}		\subfigure{\label{fig:bound_4}\includegraphics[width=0.45\columnwidth]{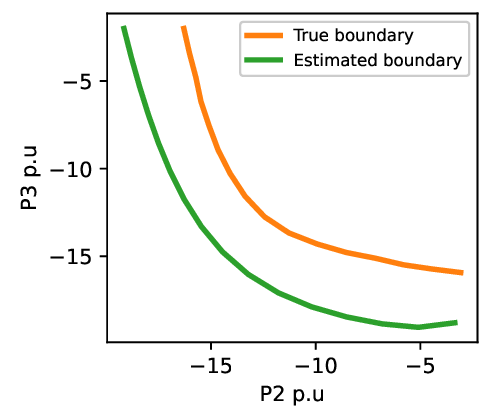}}
    \subfigure{\label{fig:bound_4_2}\includegraphics[width=0.45\columnwidth]{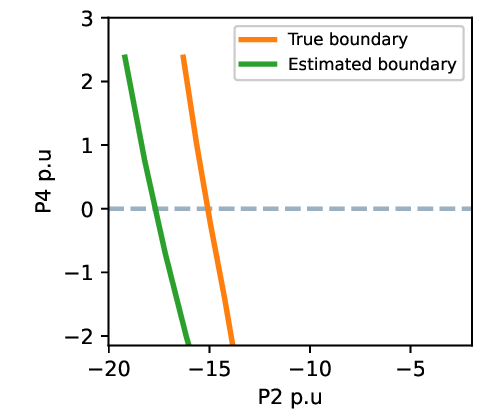}}
	\end{center}
	\caption{Voltage stability boundary of the 4-bus system.}
	\label{fig:bound4}
\end{figure}
In the 4-bus system, Fig.~\ref{fig:result_41} shows the associated manifold when the power injections vary at Bus 2 and Bus 3. Fig.~\ref{fig:result_42} shows the associated manifold when the power injections vary at Bus 2 and Bus 4. 
The orange dots represent the true manifold calculated from the continuation power flow. The green curves denote the outcomes of our proposed method. The blue stars are the locally collected data we used, which is zoomed in in Fig.~\ref{fig:data4}. Fig.~\ref{fig:bound4} shows the voltage stability boundary we estimate by the poles of the Padé approximation in green, while the true stability boundary obtained by continuation is shown in orange.

In the 9-bus system, Fig.~\ref{fig:result_91} shows the associated manifold when the power injections vary at Bus 5 and Bus 7. Fig.~\ref{fig:result_92} shows the associated manifold when the power injections vary at Bus 7 and Bus 9. 
The blue stars in Fig.~\ref{fig:data9} are the discrete local data points used to generate the associated manifold and the boundary. 
Fig.~\ref{fig:bound9} shows the voltage stability boundary we estimate by the poles of the Padé approximation. The color settings in the figures are the same as the 4-bus system plots.

\begin{figure}[tb!]
	\begin{center}	
    \subfigure[]{\label{fig:result_91}\includegraphics[width=0.45\columnwidth]{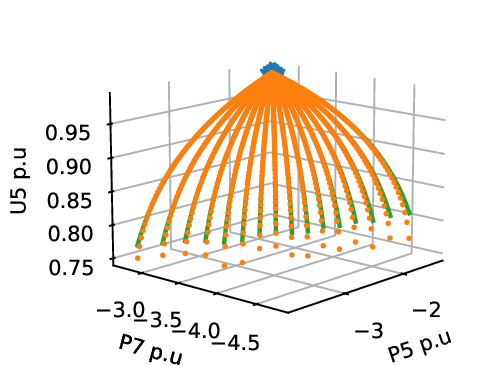}}
        \subfigure[]{\label{fig:result_92}\includegraphics[width=0.45\columnwidth]{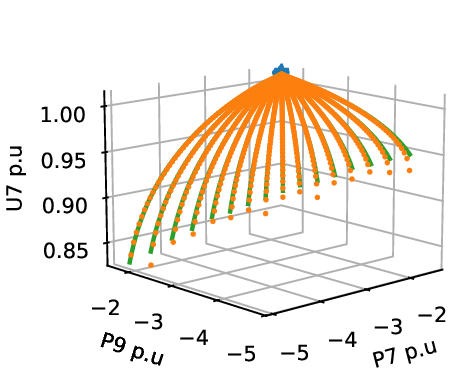}}
	\end{center}
	\caption{Manifold constructions for the 9-bus system (Orange dots: true
manifold, green curves: proposed estimation, 
blue stars: data points used).}
	\label{fig:result9}
\end{figure}

\begin{figure}[tb!]
	\begin{center}
		\subfigure{\label{fig:data_9}\includegraphics[width=0.45\columnwidth]{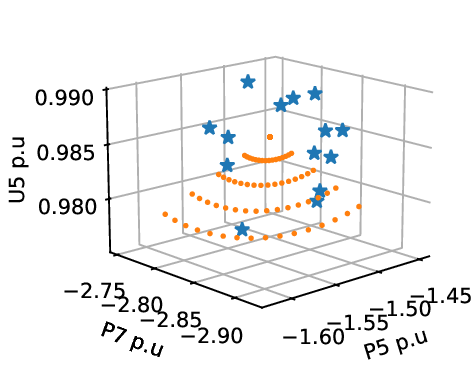}}
        \subfigure{\label{fig:data_9_2}\includegraphics[width=0.45\columnwidth]{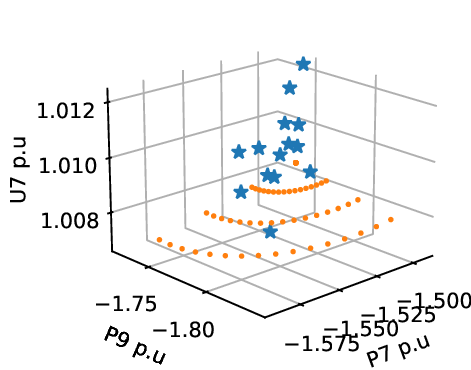}}
	\end{center}
	\caption{Data points used in the 9-bus system (blue stars).}
	\label{fig:data9}
\end{figure}

\begin{figure}[tb!]
	\begin{center}		\subfigure{\label{fig:bound_9}\includegraphics[width=0.45\columnwidth]{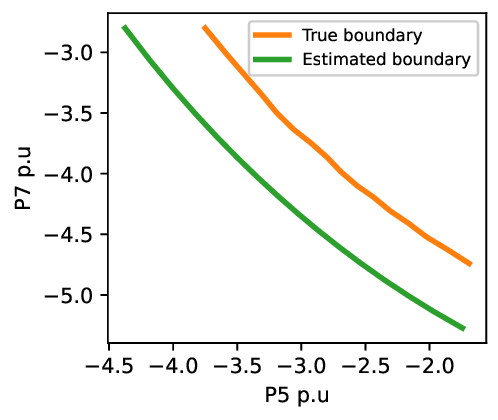}}
    \subfigure{\label{fig:bound_9_2}\includegraphics[width=0.45\columnwidth]{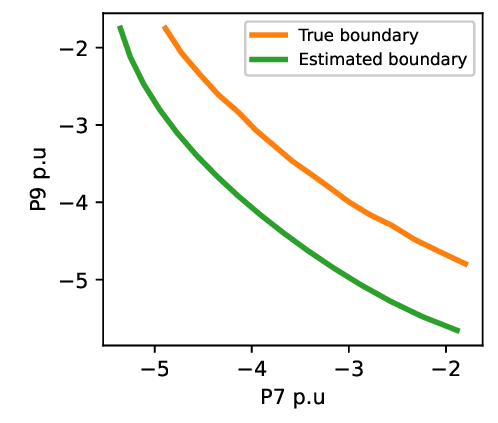}}
	\end{center}
	\caption{Voltage stability boundary of the 9-bus system.}
	\label{fig:bound9}
\end{figure}

As shown in the plots for both cases, the proposed method replicates a wide area of the associated manifold with high fidelity in every direction. Figs.~\ref{fig:result4} and \ref{fig:result9} indicate that the geodesic method demonstrates strong applicability across different systems and at diverse operating points. Figs.~\ref{fig:data4} and \ref{fig:data9} show that the method can achieve high accuracy with a small number of randomly distributed data points within a small local radial patch.
Figs.~\ref{fig:bound4} and \ref{fig:bound9} show the estimated singular boundaries for both systems. At the chosen order of the Padé approximation, the proposed method yields an estimated singular boundary that replicates the shape of the true boundary. Whether the proposed method results in an overestimation or an underestimation would depend on the order of the Padé approximation. 
Considering its model-free nature and the small number of data points used, the proposed method is shown to be a very rigorous and efficient approach. 
\subsection{Impact of Data Ranges on Approximation Errors}

\begin{figure}[tb!]
	\begin{center}		\subfigure[0.01 p.u. patch range]{\label{fig:case4_0.01}\includegraphics[width=0.45\columnwidth]{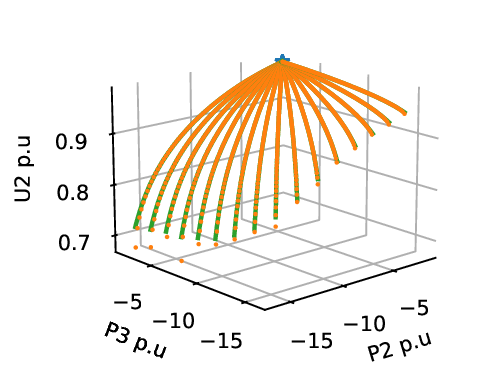}}
    \subfigure[0.05 p.u. patch range]{\label{fig:case4_0.05}\includegraphics[width=0.45\columnwidth]{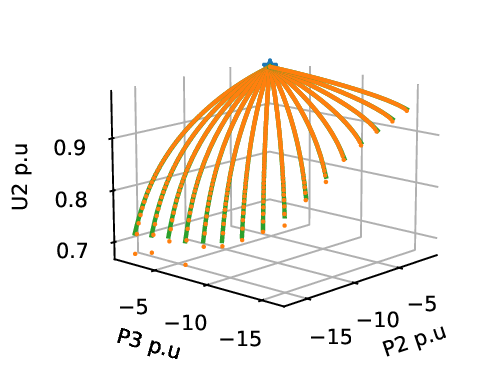}}
	\end{center}
    \begin{center}	
    \subfigure[2 p.u. patch range]{\label{fig:case4_l2}\includegraphics[width=0.45\columnwidth]{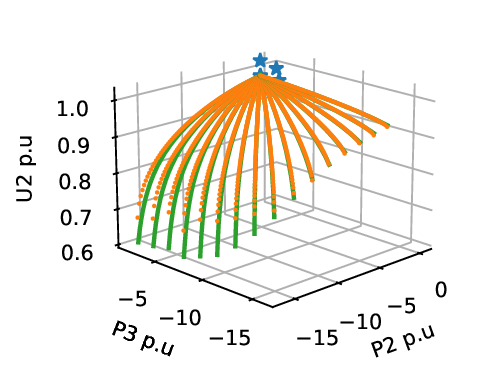}}
    \subfigure[Calculated by model]{\label{fig:case4_og}\includegraphics[width=0.45\columnwidth]{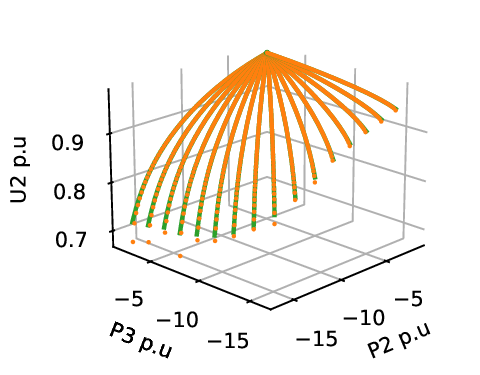}}
	\end{center}
	\caption{Manifold constructions via different patch ranges  for the 4-bus system.}
	\label{fig:length4}
\end{figure}

\begin{figure}[tb!]
	\begin{center}		\subfigure[0.01 p.u. patch range]{\label{fig:case9_0.01}\includegraphics[width=0.45\columnwidth]{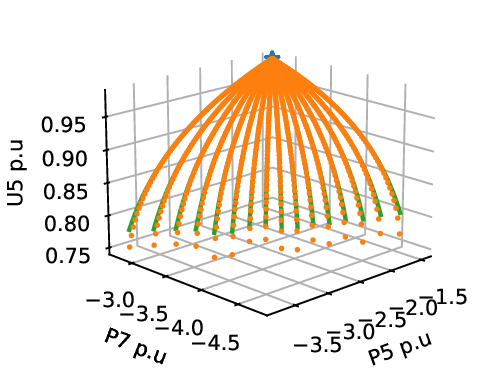}}
    \subfigure[0.05 p.u. patch range]{\label{fig:case9_0.05}\includegraphics[width=0.45\columnwidth]{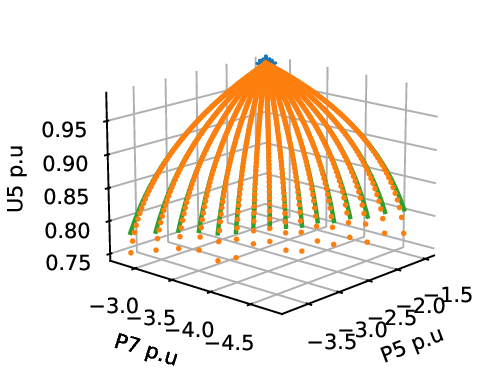}}
	\end{center}
    \begin{center}	
    \subfigure[0.5 p.u. patch range]{\label{fig:case9_0.5}\includegraphics[width=0.45\columnwidth]{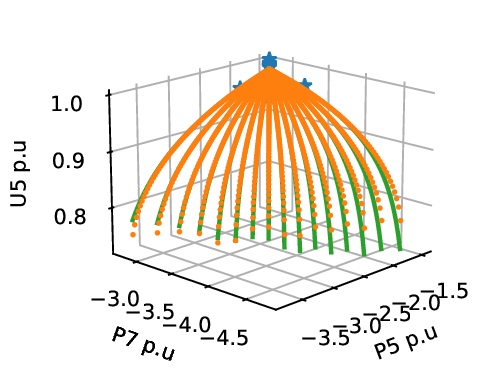}}
    \subfigure[Calculated by model]{\label{fig:case9_og}\includegraphics[width=0.45\columnwidth]{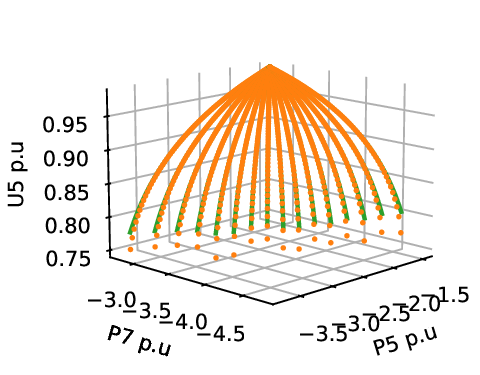}}
	\end{center}
	\caption{Manifold constructions via different patch ranges for the 9-bus system.}
	\label{fig:length9}
\end{figure}

\begin{figure}[tb!]
	\begin{center}		\subfigure[0.01 p.u. patch range]{\label{fig:case4e_0.01}\includegraphics[width=0.45\columnwidth]{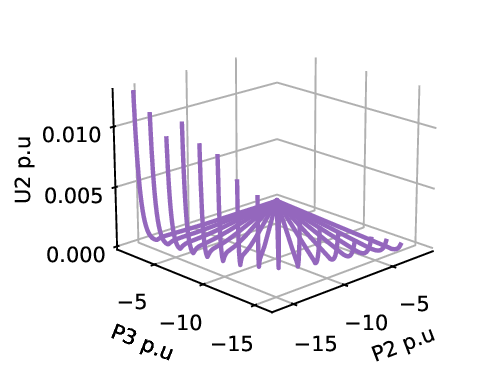}}
    \subfigure[0.05 p.u. patch range]{\label{fig:case4e_0.05}\includegraphics[width=0.45\columnwidth]{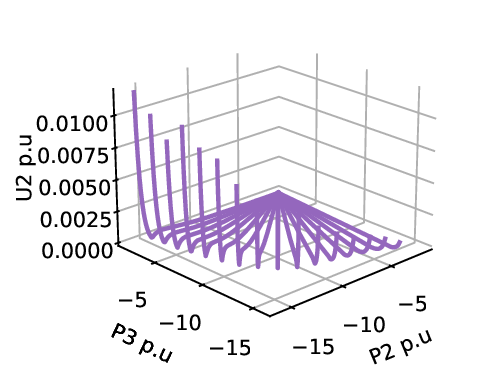}}
	\end{center}
    \begin{center}
    \subfigure[2 p.u. patch range]{\label{fig:case4e_2}\includegraphics[width=0.45\columnwidth]{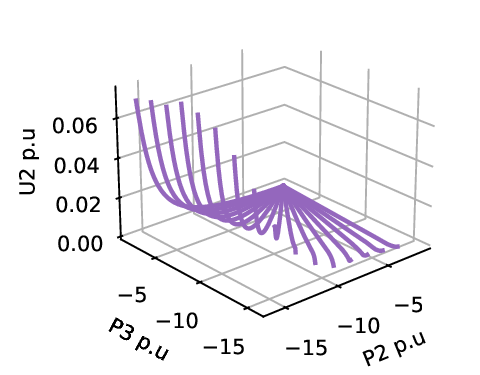}}
    \subfigure[Calculated by model]{\label{fig:case4e_og}\includegraphics[width=0.45\columnwidth]{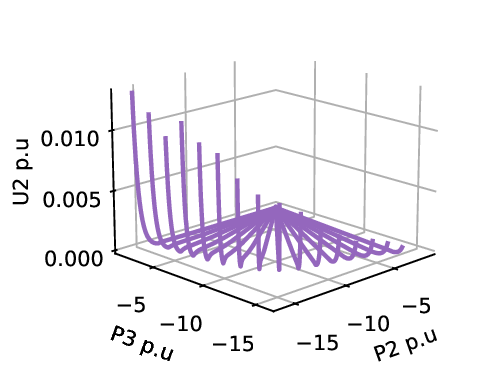}}
	\end{center}
	\caption{Error comparison of different patch ranges for the 4-bus system.}
	\label{fig:lengthe4}
\end{figure}

\begin{figure}[tb!]
	\begin{center}		\subfigure[0.01 p.u. patch range]{\label{fig:case9e_0.01}\includegraphics[width=0.45\columnwidth]{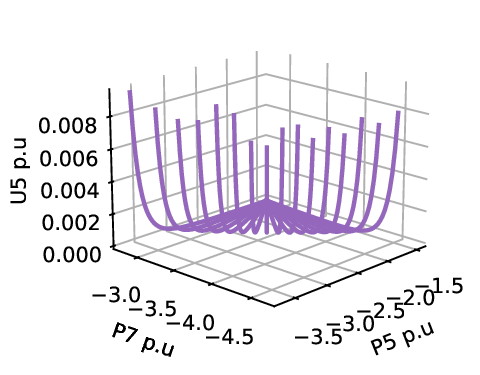}}
    \subfigure[0.05 p.u. patch range]{\label{fig:case9e_0.05}\includegraphics[width=0.45\columnwidth]{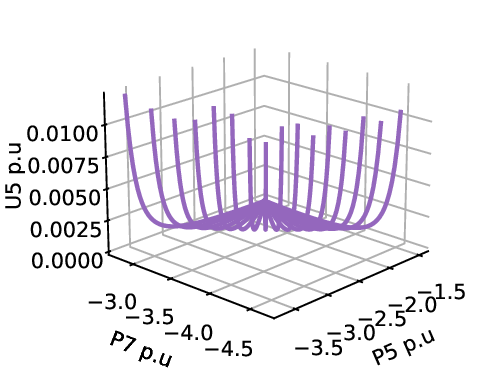}}
	\end{center}
    \begin{center}
    \subfigure[0.5 p.u. patch range]{\label{fig:case9e_0.5}\includegraphics[width=0.45\columnwidth]{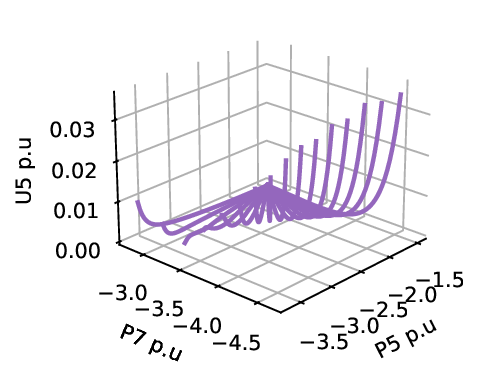}}
    \subfigure[Calculated by model]{\label{fig:case9e_og}\includegraphics[width=0.45\columnwidth]{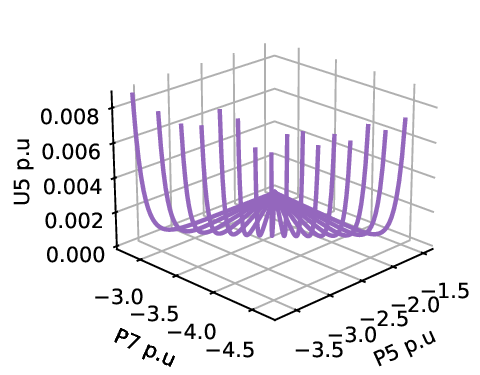}}
	\end{center}
	\caption{Error comparison of different patch ranges for the 9-bus system.}
	\label{fig:lengthe9}
\end{figure}

First, we fix the required number of power variation directions. For the 4-bus system, four independent directions are selected. For the 9-bus system, 13 directions are selected. Then, we vary the radius of the data patch from a small value at $0.01$ p.u. and continuously increase the radius. The required data points are collected on the boundary of the patch as the radius grows. Finally, we apply the proposed data-based approach to these data points and perform the manifold approximation. The respective manifolds of the two systems are shown in Fig.~\ref{fig:length4} and \ref{fig:length9}. Their respective errors are depicted in Fig.~\ref{fig:lengthe4} and \ref{fig:lengthe9}. The average approximation errors are summarized in Table~\ref{tab:table1} and \ref{tab:table3}; while the maximum approximation errors are listed in Table~\ref{tab:table2} and \ref{tab:table4}. 

As observed from the aforementioned figures and tables, the proposed data-based method maintains high accuracy across $80\%$ (even $95\%$ in some cases) of the stability region. This is particularly true when the patch radius is less than $0.05$ p.u., which corresponds to a power variation of $5$ MVA under the 100MVA base. 

To reveal the main contributor of these errors, we include the results of the model-based approximation, namely, calculating the Jacobian matrix not from the data formula \eqref{eq:ca_partial} but from the power flow equation \eqref{eq:cpq} and \eqref{eq:basis}. The corresponding visualizations can be found in the sub-figure (d) of Fig.~\ref{fig:length4} to Fig.~\ref{fig:lengthe9}; while the error values are presented in the last row of Table~\ref{tab:table1} to \ref{tab:table4}. These results illustrate that the magnitudes of errors from the data-based approximation remain the same as the magnitudes of errors from the model-based approximation. It suggests that the main contributor of these errors comes from the limited truncation of the Taylor (thus, Padé) approximation. 

\begin{table}[tb!]
	\begin{center}
		\caption{ Average voltage error in the 4-bus system}
		\label{tab:table1}
		\begin{tabular}{cccc}
			\Xhline{1.5pt}
            \multirow{2}{*}{\makecell[c]{Radius \\ of patch}} 
& \multicolumn{3}{c}{\makecell[c]{Percentage distance to boundary}} \\
& 80\% & 95\% & 99\%
\\ \Xhline{1pt} 
			\multirow{2}{*}{\makecell[c]{0.01}} & 
            \multirow{2}{*}{\makecell[c]{$1.7\times 10^{-4}$}} & \multirow{2}{*}{\makecell[c]{$1.9\times 10^{-3}$}} & \multirow{2}{*}{\makecell[c]{$1.2\times 10^{-2}$}} \\ 
                & & & \\
			\multirow{2}{*}{\makecell[c]{0.05}} & 
            \multirow{2}{*}{\makecell[c]{$3.2\times 10^{-4}$}} & \multirow{2}{*}{\makecell[c]{$1.5\times 10^{-3}$}} & \multirow{2}{*}{\makecell[c]{$1.2\times 10^{-2}$}} \\ 
                & & & \\
                \multirow{2}{*}{\makecell[c]{2.0}} & 
            \multirow{2}{*}{\makecell[c]{$8.8\times 10^{-3}$}} & \multirow{2}{*}{\makecell[c]{$2.7\times 10^{-2}$}} & \multirow{2}{*}{\makecell[c]{$3.8\times 10^{-2}$}} \\ 
                & & & \\
                \multirow{2}{*}{\makecell[c]{By Model}} & 
             \multirow{2}{*}{\makecell[c]{$1.4\times 10^{-4}$}} & \multirow{2}{*}{\makecell[c]{$2.0\times 10^{-3}$}} & \multirow{2}{*}{\makecell[c]{$1.3\times 10^{-2}$}} \\  
                & & & \\
			\Xhline{1.5pt}
		\end{tabular}
	\end{center}
\end{table}

\begin{table}[tb!]
	\begin{center}
		\caption{ Max voltage error in the 4-bus system}
		\label{tab:table2}
		\begin{tabular}{cccc}
			\Xhline{1.5pt}
            \multirow{2}{*}{\makecell[c]{Radius \\ of patch}} 
& \multicolumn{3}{c}{\makecell[c]{Percentage distance to boundary}} \\
& 80\% & 95\% & 99\%
\\ \Xhline{1pt} 
			\multirow{2}{*}{\makecell[c]{0.01}} & 
            \multirow{2}{*}{\makecell[c]{$3.6\times 10^{-4}$}} & \multirow{2}{*}{\makecell[c]{$4.7\times 10^{-3}$}} & \multirow{2}{*}{\makecell[c]{$2.4\times 10^{-2}$}} \\ 
                & & & \\
			\multirow{2}{*}{\makecell[c]{0.05}} & 
            \multirow{2}{*}{\makecell[c]{$5.6\times 10^{-4}$}} & \multirow{2}{*}{\makecell[c]{$3.9\times 10^{-3}$}} & \multirow{2}{*}{\makecell[c]{$2.3\times 10^{-2}$}}  \\ 
                & & & \\
                \multirow{2}{*}{\makecell[c]{2.00}} & 
            \multirow{2}{*}{\makecell[c]{$1.6\times 10^{-2}$}} & \multirow{2}{*}{\makecell[c]{$1.6\times 10^{-2}$}} & \multirow{2}{*}{\makecell[c]{$8.3\times 10^{-2}$}} \\ 
                & & & \\
                \multirow{2}{*}{\makecell[c]{By Model}} & 
             \multirow{2}{*}{\makecell[c]{$3.0\times 10^{-4}$}} & \multirow{2}{*}{\makecell[c]{$4.9\times 10^{-3}$}} & \multirow{2}{*}{\makecell[c]{$2.5\times 10^{-2}$}} \\ 
                & & & \\
			\Xhline{1.5pt}
		\end{tabular}
	\end{center}
\end{table}

\begin{table}[tb!]
	\begin{center}
		\caption{ Average voltage error in the 9-bus system}
		\label{tab:table3}
		\begin{tabular}{cccc}
			\Xhline{1.5pt}
            \multirow{2}{*}{\makecell[c]{Radius \\ of patch}} 
& \multicolumn{3}{c}{\makecell[c]{Percentage distance to boundary}} \\
& 80\% & 95\% & 99\%
\\ \Xhline{1pt} 
			\multirow{2}{*}{\makecell[c]{0.01}} & 
            \multirow{2}{*}{\makecell[c]{$8.2\times 10^{-4}$}} & \multirow{2}{*}{\makecell[c]{$9.3\times 10^{-3}$}} & \multirow{2}{*}{\makecell[c]{$2.9\times 10^{-2}$}} \\ 
                & & & \\
			\multirow{2}{*}{\makecell[c]{0.05}} & 
            \multirow{2}{*}{\makecell[c]{$1.9\times 10^{-3}$}} & \multirow{2}{*}{\makecell[c]{$1.2\times 10^{-2}$}} & \multirow{2}{*}{\makecell[c]{$3.2\times 10^{-2}$}} \\ 
                & & & \\
                \multirow{2}{*}{\makecell[c]{0.50}} & 
            \multirow{2}{*}{\makecell[c]{$5.1\times 10^{-3}$}} & \multirow{2}{*}{\makecell[c]{$2.1\times 10^{-2}$}} & \multirow{2}{*}{\makecell[c]{$2.6\times 10^{-2}$}} \\ 
                & & & \\
                \multirow{2}{*}{\makecell[c]{By Model}} & 
             \multirow{2}{*}{\makecell[c]{$5.4\times 10^{-4}$}} & \multirow{2}{*}{\makecell[c]{$8.6\times 10^{-3}$}} & \multirow{2}{*}{\makecell[c]{$2.8\times 10^{-2}$}} \\  
                & & & \\
			\Xhline{1.5pt}
		\end{tabular}
	\end{center}
\end{table}

\begin{table}[tb!]
	\begin{center}
		\caption{ Max voltage error in the 9-bus system}
		\label{tab:table4}
		\begin{tabular}{cccc}
			\Xhline{1.5pt}
            \multirow{2}{*}{\makecell[c]{Radius \\ of patch}} 
& \multicolumn{3}{c}{\makecell[c]{Percentage distance to boundary}} \\
& 80\% & 95\% & 99\%
\\ \Xhline{1pt} 
			\multirow{2}{*}{\makecell[c]{0.01}} & 
            \multirow{2}{*}{\makecell[c]{$9.9\times 10^{-4}$}} & \multirow{2}{*}{\makecell[c]{$1.1\times 10^{-2}$}} & \multirow{2}{*}{\makecell[c]{$3.8\times 10^{-2}$}} \\ 
                & & & \\
			\multirow{2}{*}{\makecell[c]{0.05}} & 
            \multirow{2}{*}{\makecell[c]{$2.2\times 10^{-3}$}} & \multirow{2}{*}{\makecell[c]{$1.5\times 10^{-2}$}} & \multirow{2}{*}{\makecell[c]{$4.1\times 10^{-2}$}}  \\ 
                & & & \\
                \multirow{2}{*}{\makecell[c]{0.50}} & 
            \multirow{2}{*}{\makecell[c]{$8.5\times 10^{-3}$}} & \multirow{2}{*}{\makecell[c]{$3.6\times 10^{-2}$}} & \multirow{2}{*}{\makecell[c]{$5.1\times 10^{-2}$}} \\ 
                & & & \\
                \multirow{2}{*}{\makecell[c]{By Model}} & 
             \multirow{2}{*}{\makecell[c]{$6.8\times 10^{-4}$}} & \multirow{2}{*}{\makecell[c]{$1.1\times 10^{-2}$}} & \multirow{2}{*}{\makecell[c]{$3.7\times 10^{-2}$}} \\ 
                & & & \\
			\Xhline{1.5pt}
		\end{tabular}
	\end{center}
\end{table}


\section{Conclusion}
This paper proposed a data-based method inspired by differential geometry that can evaluate any point in the high-dimensional power flow manifold from a limited number of local measurements around a single operating point. 
We first establish a theoretical framework based on differential geometry to characterize the inverse map of a bijection $r$ from manifolds $X$ to $Y$. The framework involves constructing geodesic equations on the manifold $Y$ with respect to $X$, deriving their solution in rational function form, and thereby providing an explicit description of the inverse map. Subsequently, we demonstrate that the classical power flow problem in electrical engineering can be formulated as such a bijection $r$ between manifolds, thereby bridging the proposed geometric framework with practical engineering applications. Then, we showed that all the higher-order partial derivatives of the power flow problem can be completely determined from the first-order partial derivatives. By combining these results, we are able to construct the analytic representation of any geodesic curve in the associated power flow manifold just from the first-order derivatives. Finally, a data-based oracle is proposed to estimate these first-order derivatives. The minimum number of data points required is ${2N_l+N_g+1}$, including the given operating point. Numerical examples in arbitrary power changing situations verified the efficacy of the proposed method. 

Future research include applying the method to large-scale power systems and the extension to hybrid HVDC-AC systems and systems with saturations.



\appendix

\subsection{Proof of Lemma~\ref{lemma:geo}}\label{append:1}
\begin{proof}
First, We prove that $\Omega_k$ is a polynomial function in the partial derivatives of $r$ whose orders are no greater than $k$, the velocity $\dot{x}^j$ for all $j$ and $g^{il}$ for all $i$ and $l$:
\begin{equation}\label{eq:polyOk}
     \Omega_k = \Omega_k \big(\bm r_{i_{11}},\bm r_{i_{21} i_{22}},\dots,\bm r_{i_{k1} \dots i_{kk}},\dot{x}^j,g^{il} \big).
\end{equation} 
We prove the statement by mathematical induction.

\textit{Base case:}
First, we show that for $k = 2$, $\Omega_k$ is a polynomial function in the partial derivatives of $r$ whose orders are no greater than $k$, the velocity $\dot{x}^j$ for all $j$ and $g^{il}$ for all $il$.

Substitute \eqref{eq:cris} into \eqref{eq:geodesic},
\begin{equation}\label{eq:taylor2}
    \Omega_2=\frac{1}{2}\ddot{x}^l =\frac{1}{2}\langle \bm r_{ij},\bm r_{k} \rangle g^{ml}\dot{x}^i\dot{x}^j, 
\end{equation}
where $\ddot{x}^m$=${d^2x^m}/{d\lambda^2}$.

Thus, the base case holds.

\textit{Inductive step:}
Assume that the statement holds for some arbitrary integer $k \ge 2$;
\begin{equation}\label{eq:taylorn}
    \Omega_k = \Omega_k \big(\bm r_{i_{11}},\bm r_{i_{21} i_{22}},\dots,\bm r_{i_{k1} \dots i_{kk}},\dot{x}^j,g^{il} \big).
\end{equation} 

We now show that the statement also holds for $k+1$.
\begin{subequations}
\begin{align}\label{eq:taylorn1}
\Omega_{k+1} =&\frac{\dot{\Omega}_k}{k+1}
  \\=& \bigg(\frac{\partial \Omega_k}{\partial \bm r_{i_1}}\frac{\partial \bm r_{i_1}}{\partial x^{i_{k+1}}}
  +\dots+\frac{\partial \Omega_k}{\partial \bm r_{i_1\dots i_k}}\frac{\partial \bm r_{i_1\dots i_k}}{\partial x^{i_{k+1}}}\nonumber
  \\&+\frac{\partial \Omega_k}{\partial g^{il}}\frac{\partial g^{il}}{\partial x^{i_{k+1}}}\bigg)\dot{x}^{i_{k+1}}+\frac{\partial \Omega_k}{\partial \dot{x}^{j}}\frac{\partial \dot{x}^{j}}{\partial \lambda}, 
  \\=& \bigg(\frac{\partial \Omega_k}{\partial \bm r_{i_1}}\bm r_{i_1i_{k+1}}
  +\dots+\frac{\partial \Omega_k}{\partial \bm r_{i_1\dots i_k}}\bm r_{i_1\dots i_{k+1}}\nonumber
  \\&+\frac{\partial \Omega_k}{\partial g^{il}}\frac{\partial g^{il}}{\partial x^{i_{k+1}}}\bigg)\dot{x}^{i_{k+1}}+\frac{\partial \Omega_k}{\partial \dot{x}^{j}}\ddot{x}^{j},
\end{align}
\end{subequations}
where the subscript $i_{k+1}$ is index variable taking values over the set of integers from $1$ to $n$.

Since $\Omega_k$ is a polynomial function, the partial derivatives of $\Omega_k$ are also polynomial functions, and according to \eqref{eq:taylor2}, $\ddot{x}^{j}$ is also a polynomial function. Thus, $\Omega_{k+1}$ is a polynomial function, 
\begin{equation}
    \Omega_{k+1} = \Omega_{k+1}\big(\bm r_{i_{11}},\bm r_{i_{21} i_{22}},\dots,\bm r_{i_{k1} \dots i_{kk}i_{k+1}},\dot{x}^j,g^{il} \big).
\end{equation}
And according to \eqref{eq:ginv}, $g^{il}$ is a rational function of $r_{i_{11}}$, $\Omega_{k+1}$ is a rational function of partial derivatives of $r$ whose orders are no greater than $k$ and by the velocity $\dot{x}^j$ for all $j$:
\begin{equation}
    \Omega_{k+1} = \Omega_{k+1}\big(\bm r_{i_{11}},\bm r_{i_{21} i_{22}},\dots,\bm r_{i_{k1} \dots i_{kk}i_{k+1}},\dot{x}^j \big).
\end{equation}
\end{proof}

\subsection{Proof of Lemma~\ref{th:derive}}\label{append:2}
\begin{proof}
    We prove the statement by mathematical induction.

\textit{Base case:}
First, we show that for the second-order partial derivatives of $(P,Q)$ with respect to $(V,\delta)$, namely, $k=2$, are determined by the first-order partial derivatives and the given operating point. 

Let's consider a general second-order term $\partial^2 y^i / \partial x^j \partial x^l$ where $y^i = P_i$ or $Q_i$, $x^j = V_j$ or $\delta_j$, $x^l = V_l$ or $\delta_l$. We have:
\\ if $i\neq j$ and $i \neq l$, 
\begin{subequations}\label{eq:2nd_derivei_j_k}
    \begin{align}
     \frac{\partial^2 P_i}{\partial V_j\partial V_l} =
    \frac{\partial^2 P_i}{\partial V_j\partial \delta_l} =
     \frac{\partial^2 P_i}{\partial \delta_j\partial\delta_l}=0
    \\
    \frac{\partial^2 Q_i}{\partial V_j\partial V_l} =
    \frac{\partial^2 Q_i}{\partial V_j\partial \delta_l} =
     \frac{\partial^2 Q_i}{\partial \delta_j\partial\delta_l}=0;
    \end{align}
\end{subequations}
if $i \neq j $ but $ j=l$, 
\begin{subequations}\label{eq:2nd_derivei_jk}
    \begin{align}
        \frac{\partial^2 P_i}{\partial \delta_j^2}&=\frac{\partial Q_i }{\partial\delta_j}
    \\
     \frac{\partial^2 Q_i}{\partial \delta_j^2}&=-\frac{\partial P_i }{\partial\delta_j}
    \\
    \frac{\partial^2 P_i}{\partial V_j^2}&=
     \frac{\partial^2 Q_i}{\partial V_j^2}=0; 
    \end{align}
\end{subequations}
if $i=l $ but $i \neq j $ (or $ i = j $ but $ i\neq l$), 
\begin{subequations}\label{eq:2nd_deriveij_k}
    \begin{align}
     &\frac{-V_iV_j\partial^2 P_i}{\partial V_i\partial V_j} = \frac{-\partial^2 P_i}{\partial \delta_i\partial\delta_j} = \frac{V_i\partial^2 Q_i}{\partial V_i\partial \delta_j} = \frac{-V_j\partial^2 Q_i}{\partial V_j\partial \delta_i} \nonumber\\
     &= \frac{\partial Q_i}{\partial \delta_j}\\
     & \frac{V_iV_j\partial^2 Q_i}{\partial \delta_i\partial\delta_j} = \frac{\partial^2 Q_i}{\partial \delta_i\partial\delta_j} = \frac{V_i\partial^2 P_i}{\partial V_i\partial \delta_j} =  \frac{-V_j\partial^2 P_i}{\partial V_j\partial \delta_i} \nonumber\\
     &= \frac{\partial P_i }{\partial\delta_j} 
    \end{align}
\end{subequations}
if $i=j=l$, 
\begin{subequations}
\begin{align}
    &{P_i} - \frac{V_i^2 \partial^2 P_i}{2\partial V_i^2} = \frac{-\partial^2 P_i}{\partial \delta_i^2} = V_i \frac{\partial^2 Q_i}{\partial V_i\partial \delta_i} = \frac{\partial Q_i}{\partial \delta _i} \\
    &\frac{V_i^2 \partial^2 Q_i}{2 \partial V_i^2} - Q_i = \frac{\partial^2 Q_i}{\partial \delta_i^2} = V_i \frac{\partial^2 P_i}{\partial V_i\partial \delta_i} = \frac{\partial P_i }{\partial\delta_i} 
\end{align}
\end{subequations}

Thus, the base case holds.
\begin{equation}\label{eq:rational_2}
    \tilde{\bm r}_{i_1i_{2}}= A_2 \tilde{\bm r}_{j_1}+B_2P_{j_2}+C_2Q_{j_3}
\end{equation}

\textit{Inductive step:}
Assume that the statement holds for some arbitrary integer $k \ge 2$,
\begin{equation}\label{eq:rational_n}
      \tilde{\bm r}_{i_1 \dots i_{k}}=A_k \tilde{\bm r}_{j_1}+B_kP_{j_2}+C_kQ_{j_3}.
\end{equation}

We now show that the statement also holds for $k+1$.
\begin{subequations}\label{eq:rational_ABCk}
\begin{align}
\tilde{\bm r}_{i_1 \dots i_{k+1}}&=\frac{\partial\tilde{\bm r}_{i_1 \dots i_{k}}}{\partial x^{i_{k+1}}}
\\&=\frac{\partial A_k}{\partial x^{i_{k+1}}}\tilde{\bm r}_{j_1}+A_k\tilde{\bm r}_{j_1i_{k+1}}+\frac{\partial B_k}{\partial x^{i_{k+1}}} P_{j_2}\nonumber
\\&+B_k\frac{\partial P_{j_2}}{\partial x^{i_{k+1}}}+\frac{\partial C_k}{\partial x^{i_{k+1}}} Q_{j_3}+C_k\frac{\partial Q_{j_3}}{\partial x^{i_{k+1}}}.
\end{align}
\end{subequations}

Because $A_k$, $B_k$, $C_k$ are the rational functions of $V_{j_4}$, ${\partial A_k}/{\partial x^{i_{k+1}}}$,${\partial B_k}/{\partial x^{i_{k+1}}}$,${\partial C_k}/{\partial x^{i_{k+1}}}$ are also the rational functions of $V_{j_4}$,
\begin{equation}
    \tilde{\bm r}_{i_1 \dots i_{k+1}}=A_k\tilde{\bm r}_{j_1i_{k+1}}+\hat{A}_{k+1}\tilde{\bm r}_{j_1}+\hat B_{k+1}P_{j_2}+\hat C_{k+1}Q_{j_3},
\end{equation}
where $\hat{A}_{k+1}$, $\hat{B}_{k+1}$, and $\hat{C}_{k+1}$ are rational functions of $V_{j_4}$.

Substitute \eqref{eq:rational_2} into \eqref{eq:rational_ABCk},
\begin{subequations}\label{eq:rational_ABCr}
\begin{align}
\tilde{\bm r}_{i_1 \dots i_{k+1}}&=(\hat{A}_{k+1}+A_kA_2)\tilde{\bm r}_{j_1}+(\hat B_{k+1}+A_kB_2)P_{j_2}\nonumber\\&+(\hat C_{k+1}+A_kC_2)Q_{j_3},  
\\&=A_{k+1}\tilde{\bm r}_{j_1}+B_{k+1}P_{j_2}+C_{k+1}Q_{j_3}.
\end{align}
\end{subequations}
\end{proof}

\bibliographystyle{ieeetr}
\bibliography{ref_article}

@article{stahl1997:convergence,
	title={The convergence of {Pad{\'e}} approximants to functions with branch points},
	author={Stahl, Herbert},
	journal={Journal of Approximation Theory},
	volume={91},
	number={2},
	pages={139--204},
	year={1997},
	publisher={New York, Academic Press.}
}

@article{stahl1989:convergence,
	title={On the convergence of generalized {Pad{\'e}} approximants},
	author={Stahl, Herbert},
	journal={Constructive Approximation},
	volume={5},
	number={1},
	pages={221--240},
	year={1989},
	publisher={Springer}
}

@inproceedings{lesieutre2015:efficient,
	title={An efficient method to locate all the load flow solutions-revisited},
	author={Lesieutre, Bernard and Wu, Dan},
	booktitle={Communication, Control, and Computing (Allerton), 2015 53rd Annual Allerton Conference on},
	pages={381--388},
	year={2015},
	organization={IEEE}
}

@ARTICLE{tri_sec,
  author={Wu, Dan and Wang, Bin and Wolter, Franz-Erich and Xie, Le},
  journal={IEEE Transactions on Power Systems}, 
  title={Tri-Sectional Approximation of the Shortest Path to Long-Term Voltage Stability Boundary With Distributed Energy Resources}, 
  year={2022},
  volume={37},
  number={6},
  pages={4720-4731},
  doi={10.1109/TPWRS.2022.3154708}}

@ARTICLE{VS_margins,
  author={Molzahn, Daniel K. and Lesieutre, Bernard C. and DeMarco, Christopher L.},
  journal={IEEE Transactions on Power Systems}, 
  title={A Sufficient Condition for Power Flow Insolvability With Applications to Voltage Stability Margins}, 
  year={2013},
  volume={28},
  number={3},
  pages={2592-2601},
  doi={10.1109/TPWRS.2012.2233765}}

@article{QV,
title = {Reactive power-voltage problem: conditions for the existence of solution and localized disturbance propagation},
journal = {International Journal of Electrical Power \& Energy Systems},
volume = {8},
number = {2},
pages = {66-74},
year = {1986},
issn = {0142-0615},
doi = {https://doi.org/10.1016/0142-0615(86)90001-3},
url = {https://www.sciencedirect.com/science/article/pii/0142061586900013},
author = {J. Thorp and D. Schulz and M. Ilić-Spong}
}

@ARTICLE{radialS,
  author={Chiang, H.-D. and Baran, M.E.},
  journal={IEEE Transactions on Circuits and Systems}, 
  title={On the existence and uniqueness of load flow solution for radial distribution power networks}, 
  year={1990},
  volume={37},
  number={3},
  pages={410-416},
  doi={10.1109/31.52734}}

@ARTICLE{sol,
  author={Lesieutre, B.C. and Sauer, P.W. and Pai, M.A.},
  journal={IEEE Transactions on Circuits and Systems I: Fundamental Theory and Applications}, 
  title={Existence of solutions for the network/load equations in power systems}, 
  year={1999},
  volume={46},
  number={8},
  pages={1003-1011},
  doi={10.1109/81.780380}}

@ARTICLE{fes_sol,
  author={Miu, K.N. and Hsiao-Dong Chiang},
  journal={IEEE Transactions on Circuits and Systems I: Fundamental Theory and Applications}, 
  title={Existence, uniqueness, and monotonic properties of the feasible power flow solution for radial three-phase distribution networks}, 
  year={2000},
  volume={47},
  number={10},
  pages={1502-1514},
  doi={10.1109/81.886980}}

@ARTICLE{Linear_App,
  author={Bolognani, Saverio and Zampieri, Sandro},
  journal={IEEE Transactions on Power Systems}, 
  title={On the Existence and Linear Approximation of the Power Flow Solution in Power Distribution Networks}, 
  year={2016},
  volume={31},
  number={1},
  pages={163-172},
  doi={10.1109/TPWRS.2015.2395452}}

@ARTICLE{Inner_App,
  author={Nguyen, Hung D. and Dvijotham, Krishnamurthy and Turitsyn, Konstantin},
  journal={IEEE Transactions on Power Systems}, 
  title={Constructing Convex Inner Approximations of Steady-State Security Regions}, 
  year={2019},
  volume={34},
  number={1},
  pages={257-267},
  doi={10.1109/TPWRS.2018.2868752}}

@ARTICLE{Feas_Sets,
  author={Lee, Dongchan and Nguyen, Hung D. and Dvijotham, Krishnamurthy and Turitsyn, Konstantin},
  journal={IEEE Transactions on Control of Network Systems}, 
  title={Convex Restriction of Power Flow Feasibility Sets}, 
  year={2019},
  volume={6},
  number={3},
  pages={1235-1245},
  doi={10.1109/TCNS.2019.2930896}}

@article{cui2019solvability,
  title={Solvability of power flow equations through existence and uniqueness of complex fixed point},
  author={Cui, Bai and Sun, Xu Andy},
  journal={arXiv preprint arXiv:1904.08855},
  year={2019}
}

@ARTICLE{DL,
  author={Zhang, Yichen and Liu, Jianzhe and Qiu, Feng and Hong, Tianqi and Yao, Rui},
  journal={Journal of Modern Power Systems and Clean Energy}, 
  title={Deep Active Learning for Solvability Prediction in Power Systems}, 
  year={2022},
  volume={10},
  number={6},
  pages={1773-1777},
  doi={10.35833/MPCE.2021.000424}}

@INPROCEEDINGS{datadrive1_b,
  author={Singh, Priya and Parida, S K and Chauhan, Baru and Choudhary, Niraj},
  booktitle={2020 21st National Power Systems Conference (NPSC)}, 
  title={Online Voltage Stability Assessment Using Artificial Neural Network considering Voltage stability indices}, 
  year={2020},
  volume={},
  number={},
  pages={1-5},
  doi={10.1109/NPSC49263.2020.9331954}}

@INPROCEEDINGS{datadrive1_b2,
  author={Samy, A.Karuppa and Venkadesan, A.},
  booktitle={2021 5th International Conference on Electrical, Electronics, Communication, Computer Technologies and Optimization Techniques (ICEECCOT)}, 
  title={Online Assessment of Voltage Stability Region using an Artificial Neural Network}, 
  year={2021},
  volume={},
  number={},
  pages={757-761},
  doi={10.1109/ICEECCOT52851.2021.9708047}}

@ARTICLE{datadrive2,
  author={Cui, Mingjian and Li, Fangxing and Cui, Hantao and Bu, Siqi and Shi, Di},
  journal={IEEE Transactions on Power Systems}, 
  title={Data-Driven Joint Voltage Stability Assessment Considering Load Uncertainty: A Variational Bayes Inference Integrated With Multi-CNNs}, 
  year={2022},
  volume={37},
  number={3},
  pages={1904-1915},
  doi={10.1109/TPWRS.2021.3111151}}

@INPROCEEDINGS{datadrive1,
  author={Ali, Arshad and Verma, Kusum},
  booktitle={2024 International Conference on Control, Computing, Communication and Materials (ICCCCM)}, 
  title={PMU Data Driven Machine Learning Regression Method for Real Time Voltage Stability Assessment}, 
  year={2024},
  volume={},
  number={},
  pages={708-713},
  doi={10.1109/ICCCCM61016.2024.11039940}}

@ARTICLE{datadrive3improve,
  author={Su, Heng-Yi and Lai, Chia-Ching},
  journal={IEEE Transactions on Industry Applications}, 
  title={Improving Online Voltage Stability Monitoring in Smart Grids: A Physics-Informed Guided Deep Learning Model}, 
  year={2025},
  volume={61},
  number={2},
  pages={2397-2409},
  doi={10.1109/TIA.2025.3529813}}

@article{WU2024110716,
title = {Approximating voltage stability boundary under high variability of renewables using differential geometry},
journal = {Electric Power Systems Research},
volume = {236},
pages = {110716},
year = {2024},
issn = {0378-7796},
doi = {https://doi.org/10.1016/j.epsr.2024.110716},
url = {https://www.sciencedirect.com/science/article/pii/S0378779624006023},
author = {Dan Wu and Franz-Erich Wolter and Sijia Geng}
}

@ARTICLE{manifold,
  author={Goodwin, Ariel and Maack, Jonathan and Sigler, Devon},
  journal={IEEE Transactions on Power Systems}, 
  title={Power Flow Geometry and Approximation}, 
  year={2025},
  volume={},
  number={},
  pages={1-12},
  doi={10.1109/TPWRS.2025.3612220}}

@ARTICLE{Holomorphic,
  author={Wu, Dan and Wang, Bin},
  journal={IEEE Access}, 
  title={Holomorphic Embedding Based Continuation Method for Identifying Multiple Power Flow Solutions}, 
  year={2019},
  volume={7},
  number={},
  pages={86843-86853},
  doi={10.1109/ACCESS.2019.2925384}}

@article{anderson2025ten,
  title={Ten challenges for mathematical modeling of the green-energy transition},
  author={Anderson, Edward and Ferris, Michael and Philpott, Andrew and Anitescu, Mihai and Cramton, Peter and Geng, Sijia and Green, Richard and Homem-de-Mello, Tito and Huber, Olivier and Lecl{\`e}re, Vincent and others},
  journal={Current Sustainable/Renewable Energy Reports},
  volume={12},
  number={1},
  pages={1--13},
  year={2025},
  publisher={Springer}
}

@phdthesis{wu2017phd,
  author  = "Wu, Dan",
  title   = "Algebraic set preserving mappings for electric power grid models and its applications",
  school  = "Department of Electrical and Computer Engineering, The University of Wisconsin-Madison",
  year    = "2017",
  address = "USA",
  month   = "June"
}

@mastersthesis{Franz_master,
author = {Wolter, Franz-Erich},
year = {1979},
school = {Free University of Berlin},
pages = {},
title = {Interior Metric, shortest paths and loops in Riemannian manifolds with not necessarily smooth boundary}
}

@phdthesis{Franz_phd,
author = {Wolter, Franz-Erich},
year = {1985},
pages = {},
school = {Technical University of Berlin},
title = {Cut loci in bordered and unbordered {Riemannian} manifolds}
}

@ARTICLE{TRANs_DL2,
  author={Su, Heng-Yi and Hong, Hsu-Hui},
  journal={IEEE Transactions on Power Systems}, 
  title={An Intelligent Data-Driven Learning Approach to Enhance Online Probabilistic Voltage Stability Margin Prediction}, 
  year={2021},
  volume={36},
  number={4},
  pages={3790-3793},
  doi={10.1109/TPWRS.2021.3067150}}

@ARTICLE{TRANs_DL3,
  author={Su, Heng-Yi and Lai, Chia-Ching},
  journal={IEEE Transactions on Industry Applications}, 
  title={Dynamic-Deep-Ensemble-Learning Scheme for Probabilistic Voltage Stability Margin Estimation to Enhance Resilient Power Grid Monitoring}, 
  year={2024},
  volume={60},
  number={2},
  pages={2065-2075},
  doi={10.1109/TIA.2023.3288857}}

\end{document}